\documentclass[a4paper,11pt]{amsart}

\usepackage{graphicx}
\usepackage{amsmath}
\usepackage{amssymb}

\newtheorem{thm}{Theorem}%[section]
\theoremstyle{definition}

\theoremstyle{remark}
\theoremstyle{plain}

\def\NN{{\mathbb N}}

\def\RR{{\mathbb R}}

\def\ZZ{{\mathbb Z}}

\def\veca{{\text{\boldmath$a$}}}

\def\vece{{\text{\boldmath$e$}}}

\def\vecm{{\text{\boldmath$m$}}}

\def\vecq{{\text{\boldmath$q$}}}

\def\vecu{{\text{\boldmath$u$}}}
\def\vecv{{\text{\boldmath$v$}}}

\def\vecw{{\text{\boldmath$w$}}}
\def\vecx{{\text{\boldmath$x$}}}

\def\vecy{{\text{\boldmath$y$}}}

\def\vecbeta{{\text{\boldmath$\beta$}}}

\def\vecnull{{\text{\boldmath$0$}}}

\def\scrB{{\mathcal B}}

\def\scrK{{\mathcal K}}
\def\scrL{{\mathcal L}}

\def\scrS{{\mathcal S}}

\def\scrV{{\mathcal V}}

\def\C{\operatorname{C{}}}

\def\Li{\operatorname{Li}}

\def\S{\operatorname{S{}}}

\def\SL{\operatorname{SL}}
\def\ASL{\operatorname{ASL}}

\def\O{\operatorname{O{}}}
\def\T{\operatorname{T{}}}

\def\sgn{\operatorname{sgn}}

\def\vol{\operatorname{vol}}

\def\SLSL{\SL(2,\ZZ)\backslash\SL(2,\RR)}

\def\Onder#1#2#3#4#5{#1 \setbox0=\hbox{$#1$}\setbox1=\hbox{$#2$}
       \dimen0=.5\wd0 \dimen1=\dimen0 \dimen2=\dp0 \dimen3=\dimen2
       \advance\dimen0 by .5\wd1 \advance\dimen0 by -#4
       \advance\dimen1 by -.5\wd1 \advance\dimen1 by -#4
       \advance\dimen2 by -#3 \advance\dimen2 by \ht1
       \advance\dimen2 by 0.3ex \advance\dimen3 by #5
        \kern-\dimen0\raisebox{-\dimen2}[0ex][\dimen3]{\box1}
       \kern\dimen1}
\newcommand{\R}{\mathbb{R}}
\newcommand{\Z}{\mathbb{Z}}

\newcommand{\sfrac}[2]{{\textstyle \frac {#1}{#2}}}
\newcommand{\col}{\: : \:}

\newcommand{\bn}{\mathbf{0}}

\title{Kinetic transport in the two-dimensional periodic Lorentz gas}
\author{Jens Marklof}
\author{Andreas Str\"ombergsson}
\address{School of Mathematics, University of Bristol,
Bristol BS8 1TW, U.K.\newline
\rule[0ex]{0ex}{0ex} \hspace{8pt}{\tt j.marklof@bristol.ac.uk}}
\address{Department of Mathematics, Box 480, Uppsala University,
SE-75106 Uppsala, Sweden\newline
\rule[0ex]{0ex}{0ex} \hspace{8pt}{\tt astrombe@math.uu.se}}
\date{31 March 2008}
\thanks{J.M.\ has been supported by EPSRC Research Grants GR/T28058/01 and GR/S87461/01, and a Philip Leverhulme Prize. A.S.\ is a Royal Swedish Academy of Sciences Research Fellow supported by a grant from the Knut and Alice Wallenberg Foundation.}

\begin{document}

\begin{abstract}
The periodic Lorentz gas describes an ensemble of non-interacting point particles in a periodic array of spherical scatterers. We have recently shown that, in the limit of small scatterer density (Boltzmann-Grad limit), the macroscopic dynamics converges to a stochastic process, whose kinetic transport equation is not the linear Boltzmann equation---in contrast to the Lorentz gas with a disordered scatterer configuration. The present paper focuses on the two-dimensional set-up, and reports an explicit, elementary formula for the collision kernel of the transport equation. 
\end{abstract}

\maketitle

One of the central challenges in kinetic theory is the derivation of {\em macroscopic} evolution equations---describing for example the dynamics of an electron gas---from the underlying fundamental {\em microscopic} laws of classical or quantum mechanics. An iconic mathematical model in this research area is the Lorentz gas \cite{Lorentz}, which describes an ensemble of non-interacting point particles in an infinite array of spherical scatterers. In the case of a disordered scatterer configuration, well known results by Gallavotti \cite{Gallavotti69}, Spohn \cite{Spohn78}, and Boldrighini, Bunimovich and Sinai \cite{Boldrighini83} show that the time evolution of a macroscopic particle cloud is governed, in the limit of small scatterer density (Boltzmann-Grad limit), by the linear Boltzmann equation. We have recently proved an analogous statement for a periodic configuration of scatterers \cite{partI}, \cite{partII}. In this case the linear Boltzmann equation fails (cf.~also Golse \cite{Golse07}), and the random flight process that emerges in the Boltzmann-Grad limit is substantially more complicated.

In the present paper we focus on the two-dimensional case, and derive explicit formulae for the collision kernels of the limiting random flight process. These include information not only of the velocity before and after the collision (as in the case of the linear Boltzmann equation), but also the path length until the next hit and the velocity thereafter. Our formulae thus generalize those for the limiting distributions of the free path length found by Dahlqvist \cite{Dahlqvist97}, Boca, Gologan and Zaharescu \cite{Boca03}, and Boca and Zaharescu \cite{Boca07}. The higher dimensional case is more difficult, and we refer the reader to \cite{partI}, \cite{partII}, \cite{partIV} for further information. The asymptotic estimates for the distribution tails in \cite{partIV} improve the bounds by Bourgain, Golse and Wennberg \cite{Bourgain98}. 

Our results on the Boltzmann-Grad limit complement classical studies in ergodic theory, where the scatterer size remains fixed. Bunimovich and Sinai \cite{Bunimovich80} showed that the dynamics of the two-dimensional periodic Lorentz gas is diffusive in the limit of large times, and satisfies a central limit theorem. They assumed that the periodic Lorentz gas has a finite horizon, i.e., the scatterers are configured in such a way that the path length between consecutive collisions is bounded. The corresponding result for infinite horizon has recently been established by Szasz and Varju \cite{Szasz07} following initial work by Bleher \cite{Bleher92}. See also the recent papers by Dolgopyat, Szasz and Varju \cite{dolgopyat}, and Melbourne and Nicol \cite{Melbourne05}, \cite{Melbourne07} for related studies of statistical properties of the two-dimensional periodic Lorentz gas. The case of higher dimensions is still open, even for models with finite horizon, cf.~Chernov \cite{Chernov94}, and Balint and Toth \cite{Balint07}. Scaling limits that are intermediate between the Boltzmann-Grad and the limit of large times have been explored by Klages and Dellago \cite{Klages00}.

\centerline{-----}

This paper is organized as follows. We begin by recalling the necessary details from \cite{partI}, \cite{partII}, then derive our main results---the explicit formulae for the collision kernels of the two-dimensional Lorentz gas in the Boltzmann-Grad limit. The final two sections study the implications for the distribution of free path lengths, and asymptotic properties of the collision kernels.

\centerline{-----}

To describe the main results of \cite{partI}, \cite{partII}, let us fix a euclidean lattice $\scrL\subset\RR^d$, and assume (without loss of generality) that its fundamental cell has volume one. We denote by $\scrK_\rho\subset\RR^d$ the complement of the set $\scrB_\rho^d + \scrL$ (the ``billiard domain''), and $\T^1(\scrK_\rho)=\scrK_\rho\times\S_1^{d-1}$ its unit tangent bundle (the ``phase space''), with $\vecq\in\scrK_\rho$ the position and $\vecv\in\S_1^{d-1}$ the velocity of the particle. Here $\scrB_\rho^d$ denotes the open ball of radius $\rho$, centered at the origin, and $\S_1^{d-1}$ the unit sphere. The dynamics of a particle in the Lorentz gas is defined as the motion with unit speed along straight lines, and specular reflection at the balls $\scrB_\rho^d+\scrL$. (We will here also permit more general scattering processes.) In order to pass to the Boltzmann-Grad limit, it is convenient to rescale the length units in such a way that the mean free path length remains constant as $\rho\to 0$. This is achieved by introducing the macroscopic coordinates $(\vecx,\vecv)=(\rho^{d-1}\vecq,\vecv)\in\T^1(\rho^{d-1}\scrK_\rho)$.

We take $(\vecx_0,\vecv_0)\in\T^1(\rho^{d-1}\scrK_\rho)$ as the initial data, and denote by $\tau_1(\vecx_0,\vecv_0)$ the time until the first collision (the free path length), $\vecv_1(\vecx_0,\vecv_0)$ the velocity thereafter, and analogously by $\tau_k(\vecx_0,\vecv_0)$ and $\vecv_k(\vecx_0,\vecv_0)$ the time between the $(k-1)$st and $k$th collision and the subsequent velocity. The following theorem proves the existence of a joint limiting distribution for random initial data in the limit $\rho\to 0$.
The statement is a variant of the more general Theorem 4.1 in \cite{partII}.

\begin{thm}\label{pathwayThm}
Fix $n\in\NN$. For any Borel probability measure $\Lambda$ on $\T^1(\RR^d)$ which is
absolutely continuous with respect to the Liouville measure $\vol_{\RR^d}\times\vol_{\S_1^{d-1}}$,
and for any bounded continuous function $f:\T^1(\RR^d)\times (\R_{>0} \times \S_1^{d-1})^n\to \RR$, we have
\begin{multline}\label{pathwayEq}
\lim_{\rho\to 0}  \int_{\T^1(\rho^{d-1}\scrK_\rho)} f\big(\vecx_0,\vecv_0, \tau_1(\vecx_0,\vecv_0),\vecv_1(\vecx_0,\vecv_0), \ldots, \tau_n(\vecx_0,\vecv_0),\vecv_n(\vecx_0,\vecv_0) \big) \, d\Lambda(\vecx_0,\vecv_0) 
\\ 
 =  \int_{\T^1(\RR^d)\times (\R_{>0}\times \S_1^{d-1})^n}  
f\big(\vecx_0,\vecv_0,\xi_1,\vecv_1,\ldots,\xi_n,\vecv_n \big)  p(\vecv_0,\xi_1,\vecv_1) 
p_{\bn,\vecbeta_{\vecv_0}^+}(\vecv_1,\xi_2,\vecv_2)  \\
\cdots p_{\bn,\vecbeta_{\vecv_{n-2}}^+}(\vecv_{n-1},\xi_n,\vecv_n)
\, d\Lambda(\vecx_0,\vecv_0)\,\prod_{k=1}^n  d\xi_k \,
d\!\vol_{\S_1^{d-1}}(\vecv_k)  .
\end{multline}
\end{thm}

The collision kernels $p(\vecv_0,\xi_1,\vecv_1)$ and $p_{\bn,\vecbeta_{\vecv_0}^+}(\vecv_1,\xi_2,\vecv_2)$ are in particular independent of the choice of $\Lambda$ and $\scrL$, and are characterized in \cite{partI}, \cite{partII} by certain measures on the homogeneous spaces $\ASL(d,\ZZ)\backslash\ASL(d,\RR)$ and $\SL(d,\ZZ)\backslash\SL(d,\RR)$, respectively. In the present paper we will show how, in dimension $d=2$, this abstract description can be turned into explicit, elementary formulae; the same task seems substantially more difficult in dimension $d>2$. 

Theorem \ref{pathwayThm} forms the key ingredient in the proof of the existence of a stochastic process that governs the particle dynamics in the limit $\rho\to 0$, see \cite{partII} for details. Specifically, a particle cloud with initial density $f_0$ evolves in time $t$ to the density $f_t$ given by 
\begin{equation}
	f_t(\vecx,\vecv)=\int_{\R_{>0}\times \S_1^{d-1}}  f(t,\vecx,\vecv,\xi,\vecv_+)\, d\xi\,
d\!\vol_{\S_1^{d-1}}(\vecv_+),
\end{equation}
where $f$ is the unique solution of the differential equation
\begin{equation} \label{FPKEQ}
	\big[ \partial_t + \vecv\cdot\nabla_\vecx - \partial_\xi \big] f(t,\vecx,\vecv,\xi,\vecv_+) 
	= \int_{\S_1^{d-1}}  f(t,\vecx,\vecv_0,0,\vecv)
p_{\vecnull,\vecbeta^+_{\vecv_{0}}}(\vecv,\xi,\vecv_+) \,
d\!\vol_{\S_1^{d-1}}(\vecv_0) 
\end{equation}
subject to the initial condition 
$f(0,\vecx,\vecv,\xi,\vecv_+)= f_0(\vecx,\vecv) p(\vecv,\xi,\vecv_+)$. 
Equation \eqref{FPKEQ} may be viewed as a substitute for the linear Boltzmann equation, which describes the macroscopic dynamics in a random, rather than periodic, configuration of scatterers \cite{Gallavotti69}, \cite{Spohn78}, \cite{Boldrighini83}.
Independently of our studies in \cite{partI}, \cite{partII}, Caglioti and Golse \cite{Caglioti07} have recently proposed an equation analogous to \eqref{FPKEQ} in dimension $d=2$, by assuming an independence hypothesis which is analogous to the statement of Theorem \ref{pathwayThm} above.

\centerline{-----}

In order to treat more general scattering processes than specular reflection, we recall the basic facts from \cite[Section 2.2]{partII}, limiting our attention to the case $d=2$. For a scatterer centered at the lattice point $\vecm\in\scrL\subset\RR^2$, we use the boundary coordinates $(\vecv,\vecw)\in\S_1^1\times\S_1^1$, where $\vecv$ is the velocity and $\vecw$ position so that $\vecq=\vecm+\rho\vecw$.
Let 
\begin{align}
\scrS_-:=\{(\vecv,\vecw)\in\S_1^1\times\S_1^1\col \vecv\cdot\vecw<0\}
\end{align}
be the set of {\em incoming data} at a given scatterer, i.e., the relative velocity and position at the time of collision. The corresponding {\em outgoing data} is parametrized by the set 
\begin{align}
\scrS_+:=\{(\vecv,\vecw)\in\S_1^1\times\S_1^1\col \vecv\cdot\vecw>0\}.
\end{align}
We define the scattering map by
\begin{align}
\Theta:\scrS_-\to\scrS_+,\qquad
(\vecv_-,\vecw_-)\mapsto(\vecv_+,\vecw_+).
\end{align}
In the case of the original Lorentz gas the scattering map is 
given by specular reflection,
\begin{align} \label{LORENTZSCATTERING}
\Theta(\vecv,\vecw)=(\vecv-2(\vecv\cdot\vecw)\vecw,\vecw) .
\end{align}

Let $\Theta_1(\vecv,\vecw)\in\S_1^1$ and 
$\Theta_2(\vecv,\vecw)\in\S_1^1$ be the projection of 
$\Theta(\vecv,\vecw)\in\S_1^1$ onto the first and second component,
respectively.
We assume in the following that 
\begin{enumerate}
	\item[(i)] the scattering map $\Theta$
is \textit{spherically symmetric}, i.e., if 
$(\vecv_+,\vecw_+)=\Theta(\vecv,\vecw)$ then
$(\vecv_+K,\vecw_+K)=\Theta(\vecv K,\vecw K)$ for all
$K\in \O(2)$;
	\item[(ii)] if $\vecw=-\vecv$ then $\vecv_+=-\vecv$;
	\item[(iii)] $\Theta:\scrS_-\to\scrS_+$ is $\C^1$ and 
for each fixed $\vecv\in\S_1^1$ the map $\vecw\mapsto\Theta_1(\vecv,\vecw)$
is a $\C^1$ diffeomorphism from $\{\vecw\in\S_1^1\col\vecv\cdot\vecw<0\}$
onto some open subset of $\S_1^1$.
\end{enumerate}

We will write $\varphi(\vecv,\vecu)\in [0,2\pi)$ for the angle between
any two vectors $\vecv,\vecu\in\R^2\setminus\{\bn\}$, measured counter-clockwise from $\vecv$ to $\vecu$.
Using the spherical symmetry 
and $\Theta_1(\vecv,-\vecv)=-\vecv$ one sees that 
there exists a constant $0\leq B_\Theta<\pi$ such that for each  
$\vecv\in\S_1^1$, the image of the diffeomorphism 
$\vecw\mapsto\Theta_1(\vecv,\vecw)$ is
\begin{align} \label{VPDEF}
\scrV_{\vecv}:=\{\vecu\in\S_1^1\col B_\Theta<\varphi(\vecv,\vecu)<2\pi-B_\Theta\}.
\end{align}
Let us write $\vecbeta_\vecv^-:\scrV_\vecv\to
\{\vecw\in\S_1^1\col\vecv\cdot\vecw<0\}$ for the inverse map.
Then $\vecbeta_\vecv^-$ is spherically symmetric in the sense that
$\vecbeta_{\vecv K}^-(\vecu K)=\vecbeta_{\vecv}^-(\vecu)K$ 
for all $\vecv\in\S_1^1$, $\vecu\in\scrV_\vecv$, $K\in\O(2)$,
and in particular $\vecbeta_\vecv^-(\vecu)$ is jointly $\C^1$ in $\vecv,\vecu$.

We also define
\begin{align}
\vecbeta^+_{\vecv}(\vecu)=\Theta_2(\vecv,\vecbeta^-_{\vecv}(\vecu))
\qquad (\vecv\in\S_1^1,\:\vecu\in\scrV_\vecv).
\end{align}
The map $\vecbeta^+$ is spherically symmetric and jointly $\C^1$ in 
$\vecv,\vecu$.
In terms of the original scattering situation, the point of our notation 
is the following:
Given any $\vecv_-,\vecv_+\in\S_1^1$, there exist
$\vecw_-,\vecw_+\in\S_1^1$ such that %
$\Theta(\vecv_-,\vecw_-)=(\vecv_+,\vecw_+)$ if and only if
$B_\Theta<\varphi(\vecv_-,\vecv_+)<2\pi-B_\Theta$, and in this case $\vecw_-$ and $\vecw_+$
are uniquely determined, as $\vecw_\pm=\vecbeta^\pm_{\vecv_-}(\vecv_+)$.
For example, in the case of specular reflection \eqref{LORENTZSCATTERING} we have $B_\Theta=0$ and 
\begin{equation}
	\vecw_+=\vecw_-=\frac{\vecv_+-\vecv_-}{\|\vecv_+-\vecv_-\|} .
\end{equation}

\begin{figure}
\begin{center}
\framebox{
\begin{minipage}{0.4\textwidth}
\unitlength0.1\textwidth
\begin{picture}(10,10)(0,0)
\put(0.5,1){\includegraphics[width=0.9\textwidth]{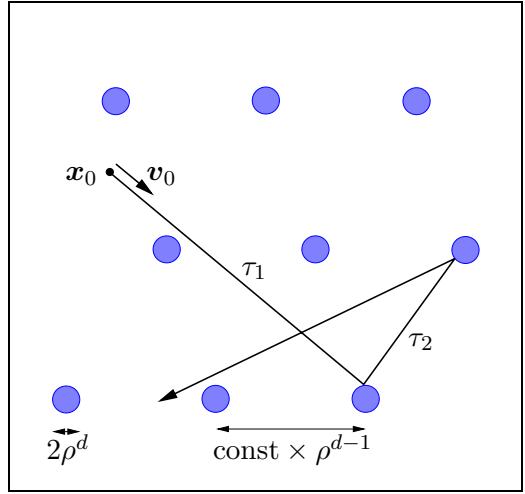}}
\put(0.8,6.4){$\vecx_0$} \put(2.5,6.4){$\vecv_0$}
\put(4.5,4.4){$\tau_1$} \put(8,2.9){$\tau_2$}
\put(0.4,0.5){$2\rho^d$} \put(3.9,0.5){$\text{const}\times\rho^{d-1}$}
\end{picture}
\end{minipage}
}
\end{center}
\caption{The periodic Lorentz gas in ``macroscopic'' coordinates ---both the lattice constant and the radius of each scatter tend to zero, in such a way that the mean free path length remains finite.}  \label{figLorentz}
\end{figure}

\begin{figure}
\begin{center}
\framebox{
\begin{minipage}{0.6\textwidth}
\unitlength0.1\textwidth
\begin{picture}(10,6)(0,0)
\put(0.5,1){\includegraphics[width=0.9\textwidth]{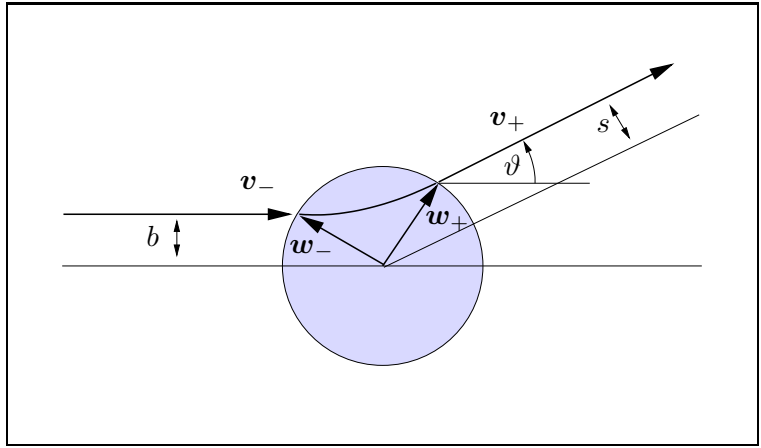}}
\put(5.6,3){$\vecw_+$}\put(3.7,2.6){$\vecw_-$}
\put(6.5,4.4){$\vecv_+$}\put(3,3.5){$\vecv_-$}
\put(6.7,3.7){$\vartheta$}\put(1.7,2.7){$b$}\put(8,4.3){$s$}
\end{picture}
\end{minipage}
}
\end{center}
\caption{Scattering in the unit ball} \label{fig}
\end{figure}

We denote the unit vectors $(1,0)$ and $(0,1)$ by $\vece_1$ and $\vece_2$, respectively. For $\vecv=(\cos\phi,\sin\phi)\in\S_1^1$ we set
\begin{equation}
	K(\vecv)=\begin{pmatrix} \cos\phi & -\sin\phi \\ \sin\phi & \cos\phi \end{pmatrix}
\end{equation}
so that $\vecv K(\vecv)=\vece_1$.
The classical impact parameter is related to the above by
\begin{equation}
	b(\vartheta)\equiv b(\vecv,\vecu)=\vecbeta_{\vece_1}^-(\vecu K(\vecv))\cdot \vece_2,
\end{equation}
where $\vartheta=\varphi(\vecv,\vecu)$.
The scattering cross section is $\sigma(\vartheta)\equiv \sigma(\vecv,\vecu)=|b'(\vartheta)|$. We define the exit parameter correspondingly by
\begin{equation}
	s(\vartheta)\equiv s(\vecu,\vecv)=\vecbeta_{\vecv K(\vecu)}^+(\vece_1) \cdot \vece_2 ,
\end{equation}
with $\vartheta=\varphi(\vecu,\vecv)$.

In the case of specular reflection \eqref{LORENTZSCATTERING}, we have $b(\vartheta)=\cos(\vartheta/2)$, $s(\vartheta)=-\cos(\vartheta/2)$. 

\centerline{-----}

We now turn to the main results of this paper.

\begin{thm} \label{PF0EXPLTHM}
For $\vecv_{0},\vecv,\vecv_+\in\S_1^1$ and $\xi>0$, the collision kernel $p_{\vecnull,\vecbeta^+_{\vecv_{0}}}$ is given by
\begin{equation}\label{formula1}
	p_{\vecnull,\vecbeta^+_{\vecv_{0}}}(\vecv,\xi,\vecv_+)
	=\sigma(\vecv,\vecv_+)\, \Phi_\vecnull\big(\xi,b(\vecv,\vecv_+),
	-s(\vecv,\vecv_0)\big)
\end{equation}
where for any $\xi>0$, $w,z\in(-1,1)$,
\begin{equation} \label{PF2DIMGENERIC}
\Phi_\vecnull(\xi,w,z)
=\frac{6}{\pi^2}
\begin{cases}
\Upsilon\Bigl(1+\frac{\xi^{-1}-\max(|w|,|z|)-1}{|w+z|}\Bigr)
& \text{if }\: w+z\neq 0
\\
0 & \text{if }\: w+z=0, \: \xi^{-1}<1+|w| \\
1 &  \text{if }\: w+z=0, \: \xi^{-1}\geq 1+|w|,
\end{cases}
\end{equation}
with
\begin{equation}
\Upsilon(x)=
\begin{cases} 
0 & \text{if }x\leq 0\\
x & \text{if }0<x<1\\
1 & \text{if }1\leq x.
\end{cases}
\end{equation}
\end{thm}

%The precise values of $\Phi_\bn(\xi,w,z)$ when $w+z=0$ are of course 
%irrelevant for the integral in the right hand side of \eqref{pathwayEq};
%the above is what we get from the definition
%

It follows that $\Phi_\bn(\xi,w,z)$ is discontinuous at each point
$(\xi,w,z)\in S$, where $S$ is the curve
$S=\bigl\{((1+|z|)^{-1},-z,z)\col -1<z<1\bigr\}$;
in fact $\Phi_\bn$ maps any neighbourhood of such a point \textit{onto}
the interval $[0,1]$.
On the other hand $\Phi_\bn$ is continuous everywhere in the complementary
region, $(\xi,w,z)\in\bigl(\R_{>0}\times(-1,1)\times(-1,1)\bigr)\setminus S$.

We now turn to the proof of Theorem \ref{PF0EXPLTHM}, which will involve computations similar to those in \cite[Sec.\ 8]{SV}.

\begin{proof}[Proof of Theorem \ref{PF0EXPLTHM}] 
Formula \eqref{formula1} follows from %\cite[Eq.~(2.27)]{partII}.
\cite[Eq.~(4.1) and Rem.\ 2.5]{partII}.
By \cite[Theorem 4.4]{partI}, we have  
\begin{align} \label{F2FORMULAPROPSTEP1}
\Phi_\vecnull(\xi,w,z)=\nu_{\vecy}\bigl(\bigl\{M\in X_1(\vecy) \col
\Z^2 M \cap R_\xi^{(z)}=\emptyset \bigr\}\bigr),
\end{align}
with $\vecy:=(\xi,z+w)$. Here $R_\xi^{(z)}$ is the open rectangle
\begin{align}
R_\xi^{(z)}=\{(x_1,x_2)\col 0<x_1<\xi,\: z-1<x_2<z+1\},
\end{align}
and 
\begin{equation}
	X_1(\vecy):=\bigl\{M\in X_1 \col \vecy\in \Z^2 M\bigr\},
	\qquad X_1:=\SLSL ,
\end{equation}
%and $\fZ(0,\xi,1)$ is the open rectangle $(0,\xi)\times(-1,1)\subset\R^2$.
cf.\ \cite[Section 7.1]{partI}.
To define the measure $\nu_\vecy$ we express $X_1(\vecy)$ as a disjoint union $X_1(\vecy)=\bigsqcup_{k=1}^\infty
X_1(k\vece_1,\vecy)$, where
\begin{equation}\label{XQKYID}
X_1(k\vece_1,\vecy) = \begin{pmatrix} k^{-1} & 0 \\ 0 & k \end{pmatrix}
\big( P^{(k)}\backslash H \big) \begin{pmatrix} \xi & z+w \\ 0 & \xi^{-1} \end{pmatrix}, 
\end{equation}
with
\begin{equation}
	H=\left\{ \begin{pmatrix} 1 & 0 \\ v & 1 \end{pmatrix} \col v\in\RR \right\},
	\qquad P^{(k)} = \left\{ \begin{pmatrix} 1 & 0 \\ k^{-2} m & 1 \end{pmatrix} \col m\in\ZZ \right\} .
\end{equation}
The measure $\nu_\vecy$ is now the Borel probability measure on $X_1(\vecy)$ defined as $(6/\pi^2)$ times the standard Lebesgue measure $dv$ on each component $X_1(k\vece_1,\vecy)$.

Note that for each $k\geq 2$ and each
$M\in X_1(k\vece_1,\vecy)$ we have $k^{-1}\vecy\in 
\Z^2 M \cap R_\xi^{(z)}$, since $z,w\in(-1,1)$.
Hence we may replace $X_1(\vecy)$ with $X_1(\vece_1,\vecy)$ in 
\eqref{F2FORMULAPROPSTEP1}.
In view of \eqref{XQKYID},
\begin{align}
X_1(\vece_1,\vecy)=\left\{
M_v:=\begin{pmatrix} \xi & z+w \\ 
v\xi & v(z+w)+\xi^{-1} \end{pmatrix} \col v\in\R/\Z\right\},
\end{align}
and hence by \eqref{F2FORMULAPROPSTEP1}, $\Phi_\vecnull(\xi,w,z)$ equals
$(6/\pi^2)$ times the Lebesgue 
measure of the set of those $v\in(0,1)$
for which the lattice spanned by 
\begin{align}
\veca_1=(\xi,z+w) \quad\text{and}\quad\veca_2=(\xi v, v(z+w)+\xi^{-1})
\end{align}
has no point in $R_\xi^{(z)}$.

\begin{figure}
\begin{center}
\framebox{\parbox{210pt}{\hspace{150pt}\\ \rule{0pt}{0pt}\hspace{5pt}
\begin{minipage}{0.4\textwidth}
\unitlength0.1\textwidth
\begin{picture}(10,9)(0,0)
\put(1.0,0.5){\includegraphics[width=0.9\textwidth]{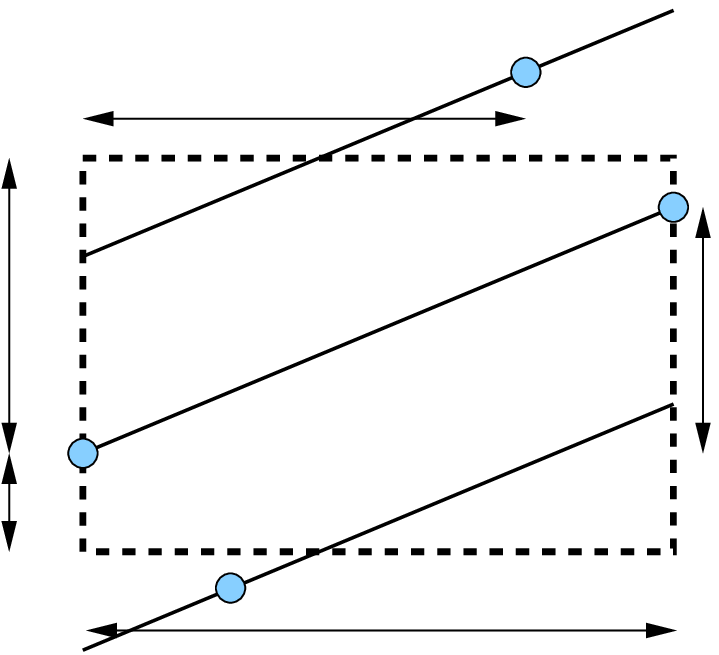}}
\put(4.0,7.6){$v\xi$} \put(3.9,4.9){$R_\xi^{(z)}$}
\put(1.38,3.1){$\bn$} \put(9.75,6.4){$\veca_1$}
\put(-0.05,4.7){$1\!\!+\!\!z$} \put(-0.05,2.2){$1\!\!-\!\!z$}
\put(10,4.6){$z\!\!+\!\!w$}
\put(4.2,1.1){$\veca_1-\veca_2$} \put(7.9,7.5){$\veca_2$}
\put(6.2,0.2){$\xi$}
\end{picture}
\end{minipage}\\
\vspace{5pt}
}}
\end{center}
\caption{The rectangle $R_\xi^{(z)}$ and points in the lattice $\Z\veca_1+\Z\veca_2$}  \label{figrect}
\end{figure}

Now for any given $v\in(0,1)$ we have $\ell\veca_1\notin R_\xi^{(z)}$
for all $\ell\in\Z$, by considering the first coordinate. It is a simple
geometric fact that for any given $\ell_1,\ell_2\in\Z$ with $|\ell_2|\geq 2$
there exists some $\ell\in\Z$ such that the point
$\ell\veca_1+(\sgn\ell_2)\veca_2$ belongs to the closed triangle with
vertices $\bn,\veca_1,\ell_1\veca_1+\ell_2\veca_2$. Thus, if
$\ell_1\veca_1+\ell_2\veca_2\in R_\xi^{(z)}$, then by convexity
$\ell\veca_1+(\sgn\ell_2)\veca_2$ lies inside $R_\xi^{(z)}$ (recall that 
$\bn$ and $\veca_1$ both lie on the boundary of $R_\xi^{(z)}$).
Therefore, if the lattice $\Z\veca_1+\Z\veca_2$ contains some point in 
$R_\xi^{(z)}$ then in fact $\ell_1\veca_1+\ell_2\veca_2\in R_\xi^{(z)}$ for
some $\ell_1\in\Z$, $\ell_2\in\{-1,1\}$.
Considering the first coordinate we also get $0<\ell_1+\ell_2v<1$, and thus
$(\ell_1,\ell_2)=(0,1)$ or $(1,-1)$.
Hence $\Phi_\vecnull(\xi,w,z)$ equals $(6/\pi^2)$ times the Lebesgue 
measure of the set of those $v\in(0,1)$
for which both $\veca_2$ and $\veca_1-\veca_2$ lie outside $R_\xi^{(z)}$.

Considering the second coordinate we see that 
$\veca_2$ and $\veca_1-\veca_2$ lie outside $R_\xi^{(z)}$ if and only if
both $v(z+w)+\xi^{-1}$ and $(1-v)(z+w)-\xi^{-1}$ lie
outside the interval $(z-1,z+1)$.
If $z+w\geq 0$ then this holds if and only if
$v(z+w)+\xi^{-1}\geq z+1$ and $(1-v)(z+w)-\xi^{-1}\leq z-1$,
or in other words, $v(z+w)\geq 1+\max(z,w)-\xi^{-1}$.
We thus obtain the formula \eqref{PF2DIMGENERIC}. 
The case $z+w<0$ is analogous.
% NOTE: We know a priori that $\Phi_\vecnull(\xi,w,z)=\Phi_\vecnull(\xi,-w,-z)$
% (cf.~\cite[Lemma 8.1]{partI}), 
%so we may assume without loss of generality that $z+w\ge 0$.
%However I don't think we need to mention this here!
\end{proof}

\begin{thm}
For $\vecv,\vecv_+\in\S_1^1$ and $\xi>0$, the collision kernel $p$ is given by
\begin{equation}\label{formula2}
	p(\vecv,\xi,\vecv_+)
	=\sigma(\vecv,\vecv_+)\, \Phi\big(\xi,b(\vecv,\vecv_+)\big)
\end{equation}
where
\begin{equation}\label{alphanull}
	\Phi(\xi,w) = \int_\xi^\infty \int_{-1}^1 \Phi_\vecnull(\eta,w,z) \, dz\, d\eta.
\end{equation}
\end{thm}

\begin{proof}
The relation \eqref{formula2} follows from \cite[Eq.~(2.27)]{partII}, and \eqref{alphanull} from \cite[Remark 6.2]{partII}.
\end{proof}

\centerline{-----}

The inner integral in \eqref{alphanull} evaluates to
\begin{equation}\label{innerin}
	\int_{-1}^1 \Phi_\vecnull(\xi,w,z) \, dz
	=
	\begin{cases}
	\frac{12}{\pi^2} & (\xi\leq \frac12) \\[5pt]
	\frac{12}{\pi^2} (\xi^{-1}-1) & \\ \qquad + \frac{6}{\pi^2} (\xi^{-1}-1+|w|) \log (\frac{1+|w|}{\xi^{-1}-1+|w|})   & \\ \qquad + \frac{6}{\pi^2} (\xi^{-1}-1-|w|) \log (\frac{1-|w|}{\xi^{-1}-1-|w|}) & (\frac12 < \xi < \frac{1}{1+|w|}) \\[5pt]
	\frac{6}{\pi^2} (\xi^{-1}-1+|w|)(1+\log(\frac{1+|w|}{2|w|})) & \\ 
	\qquad + \frac{6}{\pi^2}(\xi^{-1}-1-|w|) \log(\frac{2|w|}{1+|w|-\xi^{-1}}) & (\frac{1}{1+|w|} < \xi < \frac{1}{1-|w|}) \\[5pt]
	0 & (\xi\geq \frac{1}{1-|w|}) .  
	\end{cases}
\end{equation}
To state the result of the second integration in \eqref{alphanull}, we 
define the auxiliary function $\Psi(a)$ for $a>0$ as the
solution to $\Psi'(a)=(a^{-1}-1)\log \bigl|a^{-1}-1\bigr|$ with 
$\Psi(1)=0$; thus
\begin{align} \label{PSIDEF}
\Psi(a)=\begin{cases}
-\Li_2(a)+(1-a)\log(a^{-1}-1)+\log a-\sfrac 12(\log a)^2+\sfrac{\pi^2}6
& a<1
\\[5pt]
\Li_2(a^{-1})-(a-1)\log(1-a^{-1})+\log a-\sfrac{\pi^2}6
%\Li_2(1-a)-(a-1-\log a)\log(1-a^{-1})+\log a+\sfrac 12(\log a)^2
& a>1,
\end{cases}
\end{align}
with the dilogarithm $\Li_2(z)=\sum_{k=1}^\infty\frac{z^k}{k^2}$ ($|z|\leq 1$).
%
% Note this is not quite what Maple calls dilogarithm function, but it should
% be ok, cf. http://mathworld.wolfram.com/Dilogarithm.html ``There are two 
% different commonly encountered normalizations...'' (however following
% up on this link shows that NONE is the same as Maple's normalization!)
% Also see http://en.wikipedia.org/wiki/Polylogarithm  
%
For $\xi$ in the range $\frac12 < \xi < \frac{1}{1+|w|}$ we set
\begin{equation} \label{FDEF}
F(\xi,w)=\Psi\bigl(\xi(1+|w|)\bigr)+\Psi\bigl(\xi(1-|w|)\bigr)-2\log\xi
+2\Bigl(1+|w|\log\Bigl(\frac{1-|w|}{1+|w|}\Bigr)\Bigr)\xi,
\end{equation}
and for $\xi$ in the range $\frac{1}{1+|w|} < \xi < \frac{1}{1-|w|}$ we set
\begin{equation} \label{GDEF}
G(\xi,w)=\Psi\bigl(\xi(1+|w|)\bigr)-\log\xi+\Bigl(1-|w|+2|w|\log\Bigl(\frac{2|w|}{1+|w|}\Bigr)\Bigr)\xi.
\end{equation}
The density $\Phi(\xi,w)$ in \eqref{alphanull} can now be written as 
\begin{equation} \label{PHIXIWFORMULA}
	\Phi(\xi,w)=
	\begin{cases}
	1- \frac{12}{\pi^2} \xi & (\xi\leq \frac12) \\[5pt]
	1 + \frac{6}{\pi^2} [F(\xi,w)- F(\tfrac12,w)-1] & (\frac12 \leq \xi \leq \frac{1}{1+|w|}) \\[5pt]
	\frac{6}{\pi^2} [G(\xi,w)-G(\frac{1}{1-|w|},w)] & (\frac{1}{1+|w|} \leq \xi \leq \frac{1}{1-|w|}) \\[5pt]
	0 & (\xi\geq \frac{1}{1-|w|}) .  
	\end{cases}
\end{equation}

\begin{figure}
\begin{center}
\includegraphics[width=0.5\textwidth,angle=270]{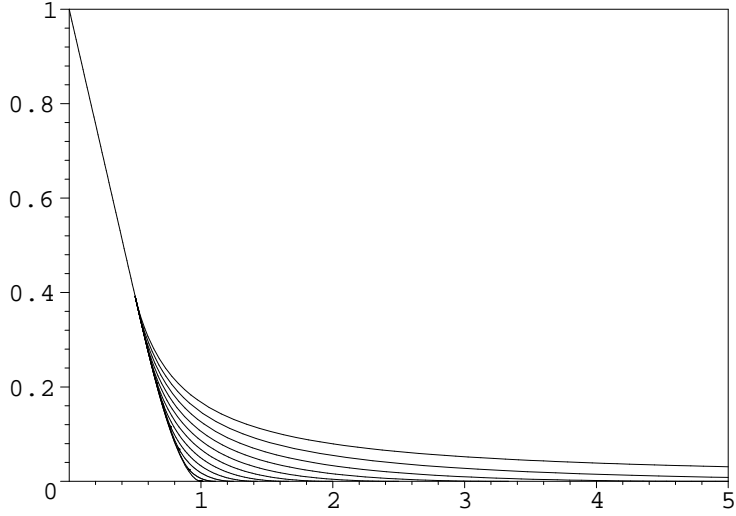}
\end{center}
\caption{Graph of $\Phi(\xi,w)$ as a function of $\xi>0$, with $w=1$, $0.9$, $0.8,\ldots,0$ (top to bottom).}  \label{figxiplot}
\end{figure}

\begin{figure}
\begin{center}
\includegraphics[width=0.5\textwidth,angle=270]{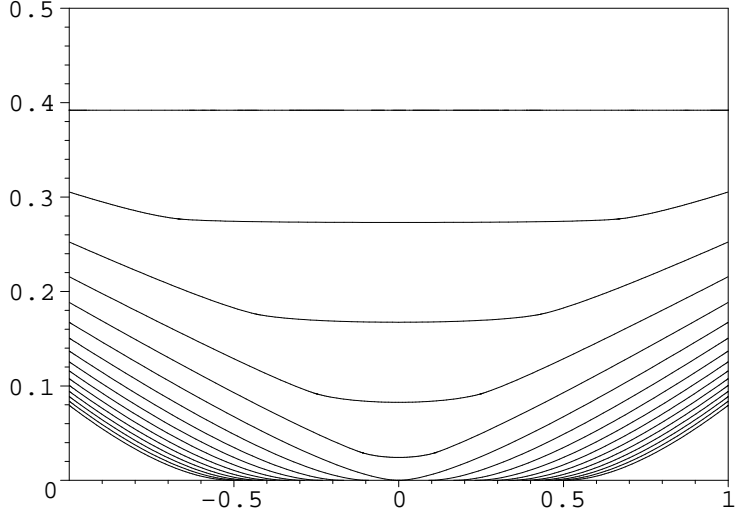}
\end{center}
\caption{Graph of $\Phi(\xi,w)$ as a function of $w\in(-1,1)$, with $\xi=0.5$, $0.6$, $0.7,\ldots,2$ (top to bottom).}  \label{figwplot}
\end{figure}

We have already seen in \cite[Remark 4.5]{partI} that $\Phi(\xi,w)$ is continuous for all $(\xi,w)\in \R_{>0}\times (-1,1)$.
The fact that the formulae on the right hand side of \eqref{PHIXIWFORMULA} agree at the points $\xi=\frac 1{1+|w|}$ can be
verified directly from \eqref{PSIDEF}--\eqref{GDEF}, 
using the identity
$\Li_2(\frac{1+|w|}2)+\Li_2(\frac{1-|w|}2)=
-\log(\frac{1+|w|}2)\log(\frac{1-|w|}2)+\frac{\pi^2}6$. 
The function $\Phi(\xi,w)$ can furthermore be continuously extended to $\R_{>0}\times [-1,1]$ by setting
\begin{equation} \label{PHIXIWFORMULA1}
	\Phi(\xi,\pm 1)=
	\begin{cases}
	1- \frac{12}{\pi^2} \xi & (\xi\leq \frac12) \\[5pt]
	1+\frac{6}{\pi^2} [\Psi(2\xi)-\log(2\xi)-1] & (\xi \geq \frac12).  
	\end{cases}
\end{equation}

We finally remark that the function $\Phi(\xi,0)$ was computed in \cite[Prop.\ 3]{SV} in a different context; set
$\Phi(\xi,0)=f_1^{\text{box},\SL_2}(4\xi)$ in the notation of that paper.

\centerline{-----}

Dahlqvist \cite{Dahlqvist97}, Boca, Gologan and Zaharescu \cite{Boca03}, and Boca and Zaharescu \cite{Boca07} have obtained explicit formulae for the limiting distribution of the free path lengths in dimension two. These can be recovered from the above expressions via the relations (cf.~\cite[Remark 4.3]{partI})
\begin{equation}\label{FP1}
	\overline\Phi_\vecnull(\xi)=\frac12 \int_{-1}^1\int_{-1}^1 \Phi_\vecnull(\xi,w,z)\, dw\, dz,
\end{equation}
for the free path length between consecutive collisions,
\begin{equation}\label{FP2}
	\Phi(\xi)=\int_{-1}^1 \Phi(\xi,w)\, dw = 2 \int_\xi^\infty \overline\Phi_\vecnull(\eta)\,d\eta,
\end{equation}
for the free path length from a generic initial point inside the billiard domain, and 
\begin{equation}\label{FP3}
	\Phi_\vecnull(\xi)=\int_{-1}^1 \Phi_\vecnull(\xi,0,z)\, dz 
\end{equation}	
for the free path length of a particle starting at a lattice point (with the scatterer removed).

\centerline{-----}

We conclude with a brief asymptotic analysis of the collision kernels, when $\xi\to\infty$ (the limit $\xi\to 0$ is trivial). We assume in the following that $\xi> 1$. A short calculation shows that, if $w z\leq 0$, we have $\Phi_\vecnull(\xi,w,z)=0$. Otherwise, for $w z> 0$, 
\begin{equation}
	\Phi_\vecnull(\xi,w,z)
=\frac{6}{\pi^2}
\Upsilon\Bigl(\frac{\xi^{-1}+\min(|w|,|z|)-1}{|w|+|z|}\Bigr)=\frac{6}{\pi^2}
\Upsilon\Bigl(\frac{1-\max(u,y)}{2\xi-(u+y)}\Bigr),
\end{equation}
where $u=\xi(1-|w|)$, $y=\xi(1-|z|)$. Thus
\begin{equation} \label{PHI0ASYMPT}
	\Phi_\vecnull(\xi,w,z)
=\frac{3}{\pi^2} 
\begin{cases}
\big(1-\max(u,y)\big) \xi^{-1} + O(\xi^{-2}) & (wz>0,\; y,u\in[0,1)) \\
0  & (\text{otherwise}) ,
\end{cases}
\end{equation}
uniformly with respect to $w,z\in(-1,1)$ as $\xi\to\infty$.
A simple integration yields for \eqref{innerin}
\begin{equation}\label{asymp1}
\int_{-1}^1 \Phi_\vecnull(\xi,w,z) \, dz =
\frac{3}{2\pi^2} 
\begin{cases}
\big(1-u^2\big) \xi^{-2} + O(\xi^{-3}) & (u\in[0,1)) \\
0  & (\text{otherwise}) .
\end{cases}	
\end{equation}
We substitute $\eta=t\xi$ in \eqref{alphanull} and apply \eqref{asymp1} to obtain
\begin{equation}\label{asymp2}
\begin{split}
	\Phi(\xi,w) & =\xi \int_1^\infty \int_{-1}^1 \Phi_\vecnull(t\xi,w,z) \, dz\, dt \\
	& = \frac{3}{2\pi^2}\, \xi \int_1^{1/u} \bigg( (1-u^2 t^2) (t\xi)^{-2} + O\big((t\xi)^{-3}\big)\bigg) \, dt \\
	& = \frac{3}{2\pi^2} (1-u)^2 \xi^{-1} +O(\xi^{-2})
\end{split}
\end{equation}
if $u\in[0,1)$, and $\Phi(\xi,w)=0$ otherwise.
Both \eqref{asymp1} and \eqref{asymp2} hold
uniformly with respect to $w\in(-1,1)$ as $\xi\to\infty$.
These asymptotics can of course also be obtained by expanding \eqref{innerin} and \eqref{PHIXIWFORMULA}, respectively.

\end{document}